\documentclass[sigconf]{acmart}
\acmConference[ICSE 2022]{The 44th International Conference on Software Engineering}{May 21–29, 2022}{Pittsburgh, PA, USA}
\usepackage{amsmath,amsfonts}
\usepackage{graphicx}
\usepackage{subfig}
\usepackage{textcomp}
\usepackage{xcolor}
\usepackage{booktabs}
\usepackage{todonotes}
\usepackage{comment}    
\usepackage{algpseudocode}
\usepackage{multirow}
\usepackage{mathtools}
\usepackage{algpseudocode}
\usepackage{amsthm} 

\newtheorem{definition}{Definition}[section]
\newtheorem{example}{Example}[section]
\newtheorem{proposition}{Proposition}[section]

\usepackage{tikz}
\usetikzlibrary{automata, positioning}
\usetikzlibrary{shapes.geometric, arrows}
\usepackage{textcomp}
\usepackage{pgfplots}

\usepackage[linesnumbered,ruled,vlined]{algorithm2e}

\SetCommentSty{mycommfont}
\SetKwRepeat{Do}{do}{while}
\SetKwInput{KwInput}{Input}                % Set the Input
\SetKwInput{KwOutput}{Output}              % set the Output

\tikzstyle{startstop} = [rectangle, rounded corners, minimum width=1cm, minimum height=0.5cm,text centered, draw=black]
\tikzstyle{io} = [trapezium, trapezium left angle=70, trapezium right angle=110, minimum width=2cm, minimum height=0.5cm, text centered, draw=black]
\tikzstyle{process} = [rectangle, minimum width=2cm, minimum height=0.5cm, text centered,draw=black]
\tikzstyle{process_n} = [rectangle, minimum width=1.5cm, minimum height=0.5cm, text centered,text width=1.5cm, draw=black]
\tikzstyle{decision} = [diamond, minimum width=1.2cm, minimum height=0.5cm, text centered, aspect=1.5, draw=black]
\tikzstyle{arrow} = [->,>=stealth]

\copyrightyear{2022} 
\acmYear{2022} 
\setcopyright{acmlicensed}\acmConference[ICSE '22]{44th International Conference on Software Engineering}{May 21--29, 2022}{Pittsburgh, PA, USA}
\acmBooktitle{44th International Conference on Software Engineering (ICSE '22), May 21--29, 2022, Pittsburgh, PA, USA}
\acmPrice{15.00}
\acmDOI{10.1145/3510003.3510080}
\acmISBN{978-1-4503-9221-1/22/05}

\begin{document}

\title{Causality-Based Neural Network Repair}

%\author{Bing Sun\inst{1}, Jun Sun\inst{1}, Hong Long Pham\inst{1}, and Jie Shi\inst{2}}
%\institute{$^1$Singapore Management University \\
%$^2$Huawei Singapore\\
%}

\author{Bing Sun}
\affiliation{%
  \institution{Singapore Management University}
}

\author{Jun Sun}
\affiliation{%
  \institution{Singapore Management University}
}

\author{Long H. Pham}
\affiliation{%
  \institution{Singapore Management University}
}

\author{Jie Shi}
\affiliation{%
  \institution{Huawei Singapore}
}

%\renewcommand{\shortauthors}{Trovato and Tobin,~\emph{et al.}}

\begin{abstract}
Neural networks have had discernible achievements in a wide range of applications. The wide-spread adoption also raises the concern of their dependability and reliability. Similar to traditional decision-making programs, neural networks can have defects that need to be repaired. The defects may cause unsafe behaviors, raise security concerns or unjust societal impacts. In this work, we address the problem of repairing a neural network for desirable properties such as fairness and the absence of backdoor. The goal is to construct a neural network that satisfies the property by (minimally) adjusting the given neural network's parameters (i.e., weights). Specifically, we propose CARE (\textbf{CA}usality-based \textbf{RE}pair), a causality-based neural network repair technique that 1) performs causality-based fault localization to identify the `guilty' neurons and 2) optimizes the parameters of the identified neurons to reduce the misbehavior. We have empirically evaluated CARE on various tasks such as backdoor removal, neural network repair for fairness and safety properties. Our experiment results show that CARE is able to repair all neural networks efficiently and effectively. For fairness repair tasks, CARE successfully improves fairness by $61.91\%$ on average. For backdoor removal tasks, CARE reduces the attack success rate from over $98\%$ to less than $1\%$. For safety property repair tasks, CARE reduces the property violation rate to less than $1\%$. Results also show that thanks to the causality-based fault localization, CARE's repair focuses on the misbehavior and preserves the accuracy of the neural networks. %Overall results suggest the feasibility of CARE of repairing neural networks for fairness.
\end{abstract}

\maketitle

\section{Introduction} \label{sec:intro}
Neural networks have had discernible achievements in a wide range of applications, ranging from medical diagnosis~\cite{medical_diagnosis}, facial recognition~\cite{face_recognition}, fraud detection~\cite{fraud_detection} and self-driving~\cite{selfdriving}. While neural networks are demonstrating excellent performance, there has been a growing concern on whether they are reliable and dependable. Similar to traditional decision-making programs, neural networks inevitably have defects and need to be repaired at times. Neural networks are usually inherently black-boxes and do not provide explanations on how and why decisions are made in certain ways. As a result, these defects are more ``hidden" compared to traditional decision-making programs~\cite{demographic,verification,minimal20}. It is thus crucial to develop systematic ways to make sure defects in a neural network are repaired and desirable properties are satisfied.

%Although fairness has always been a ripe topic for philosophical debate, there are no universally agreed definitions~\cite{fairsquare}. A number of formal definitions of fairness have been proposed within the context of automated decision-making programs. To name only a few, individual fairness dictates that inputs that differ only by certain sensitive feature (such as gender or race) should be labeled with the same prediction, whereas group fairness demands that a minority class of inputs should have similar probability of being predicted with certain (favorable) label to that of the majority class. In this work, we focus on group fairness for its relevance in a wide variety of neural network applications. We refer the readers to~\cite{science} for detailed definitions of fairness.

Similar to the activity known as debugging for traditional software programs, there is often a need to modify a neural network to fix certain aspects of its behavior (whilst maintaining other functionalities). Existing efforts to repair the unexpected behavior of neural networks often focus on retraining with additional data~\cite{deep_mutation,mode}. Although retraining is natural and often effective, retraining a neural network model could be difficult and expensive for real-world applications~\cite{minimal20}. More importantly, unlike debugging traditional software program (where we can be reasonably certain that the bug is eliminated after the `fix'), there is no guarantee that the retrained model eliminates the unwanted behavior. Therefore, techniques for modifying an existing model without retraining will be preferable in certain scenarios. That is, we sometimes would like to apply a small modification on an existing neural network to remove unexpected behaviors whilst retaining most of its well-trained behaviors.

Neural networks may misbehave in different ways. Given a model trained primarily for accuracy, it could be discriminative, i.e., the prediction is more favourable to certain groups thus violating fairness property. In this situation, small adjustment can be applied to the trained network without retraining, i.e., to improve its fairness whilst maintaining the model's accuracy. In another scenario, malicious hidden functionalities (backdoor) could be easily embedded into a neural network if the training data is poisoned~\cite{trojanNN,NC}. Such a backdoor produces unexpected behavior if a specific trigger is added to an input, whilst presenting normal behavior on clean inputs. In this scenario, one would like to repair the neural network by removing the backdoor while maintaining the correct behavior on clean inputs. A further scenario is that a trained model could violate safety-critical properties on certain inputs, e.g., the output is not within the expected range~\cite{acas,Reluplex}. In this case, we would like to repair the network by making adjustments to its parameters so that the repaired model satisfies the specified property.

For both traditional software programs and neural networks, debugging and repairing are essentially reasoning over causality. For traditional software programs, the notion of causality is natural and well defined~\cite{causal_test,wy_causal,causal_canvas}, e.g., based on (control and data) dependency. However, the same is not true for neural networks, i.e., a wrong decision could easily be the collective result of all neurons in the network and yet considering that all of them are `responsible' and should be modified is unlikely to be helpful. Although causality analysis has been applied to interpret machine learning models~\cite{cal_att} and verification~\cite{causal_learning}, existing causality analysis typically focuses on the causal relation from the input to the model prediction, and not the hidden neurons. Hence it is still not clear how to apply existing definitions of causality for repairing neural networks. In other words, how to repair a neural network to remedy defects based on causality is still largely an open problem.

In this work, we introduce CARE (\textbf{CA}usality-based \textbf{RE}pair), a general automatic repair algorithm for neural networks. Instead of retraining, CARE applies minor modifications to a given model's parameters (i.e., weights) to repair known defects whilst maintaining the model's accuracy. The defects are defined based on desirable properties that the given neural network fails to satisfy. CARE is a search-based automated program repair technique that first performs fault localization based on causal attribution of each neuron, i.e., it locates those neurons that have the most contribution to the model's defects. Secondly, CARE performs automatic repair by adjusting model parameters related to the identified neurons until the resultant model satisfies the specified property (whilst maintaining accuracy). 

%Secondly, CARE performs automatic repair in three steps, i.e., it 1) samples a set of inputs based on the statistical distribution of the given dataset, 2) generates repair candidates by adjusting model parameters related to the identified neurons, and 3) verifies that the resultant model satisfy  (and maintains certain level of accuracy). 
We summarize our contributinos as follows.
\begin{itemize}
    \item We propose and implement CARE as a general automatic repair algorithm, that applies causality-based fault localization and PSO optimization to all layers of neural networks.
    \item We demonstrate the effectiveness of CARE in the context of three different tasks: 1) fairness improvement, 2) backdoor removal and 3) safety property violation repair.
    \item We empirically evaluate CARE on multiple neural networks including feed-forward neural networks (FFNNs) and convolutional neural networks (CNNs), trained on multiple benchmark datasets. The results indicate that CARE improves the models' performance over the specified properties significantly and CARE outperforms existing approaches proposed for neural network repair. 
\end{itemize}

%We have empirically evaluated CARE on multiple neural network models, including feed-forward neural networks (FFNNs) and convolutional neural networks (CNNs), trained on multiple benchmark datasets. % that are subjects of multiple previous studies on neural network verification, fairness testing and poisoning attack~\cite{adf,provable,Reluplex,nnrepair}. 
%The experiment results show that CARE successfully repairs all experimented models by adjusting their model parameters (i.e., weights) based on the fault localization results, and improves the models' performance over the specified properties significantly. The experiment results also show that CARE's fault localization is effective when it is compared with the performance of alternative approaches, i.e., repairing by optimizing randomly selected parameters, optimizing all parameters (i.e., without any localization) and optimizing neurons selected based on gradients. Lastly, our result shows that CARE outperforms existing approaches proposed for neural network repair. 

%To the best of our knowledge, CARE is the first approach on reparing neural networks' against fairness properties. 
The remainder of this paper is organized as follows. In Section~\ref{sec:preli}, we review relevant background and define our problem. In Section~\ref{sec:approach}, we present each step of our approach in detail. In Section~\ref{sec:imple}, we evaluate our approach through multiple experiments to answer multiple research questions. We review related work in Section~\ref{sec:related} and conclude in Section~\ref{sec:conclusion}. 

\section{Preliminary} \label{sec:preli}
In this section, we review relevant background and define our research problem.

\subsection{Neural Network Properties}
There are many desirable properties we would like a trained neural network to satisfy besides meeting its accuracy requirement. In this work, we assume $\phi$ to be one of the following three kinds of properties and show that CARE can handle all of them. \\

\noindent \emph{Fairness}: Fairness is a desirable property that potentially has significant societal impact~\cite{trust_ai}. Due to the fact that machine learning models are data-driven and the training data could be biased (or imbalanced), models trained based on such data could be discriminative~\cite{fairtest,fairsquare}. In this work, we define independence-based fairness following~\cite{FM21} as follows.

\begin{definition}[Independence-based Fairness]
\label{def:fairness}
Let $N$ be a neural network and $\xi$ be a positive real-value constant. We write $Y$ as the prediction of $N$ on a set of input features $X$ and $L$ as the prediction set. We further write $F \subseteq X$ as a feature encoding some protected characteristics such as gender, age and race. $N$ satisfies independence-based fairness, with respect to $\xi$, if and only if, $\forall l \in L \: \forall f_i, f_j \in F~such~that~i \neq j$,
\begin{equation}
\begin{aligned}
| \: P(Y=l \: | \: F=f_i)-P(Y=l \: | \: F=f_j) \: | \: \leq \xi
\end{aligned}
\end{equation}

\end{definition}

Intuitively, Definition~\ref{def:fairness} states that, $N$ is fair as long as the probability difference of a favourable prediction for instances with different values of protected feature is within the threshold $\xi$.\\

\noindent \emph{Absence of backdoor}: With the wide adoption of neural networks in critical decision-making systems, sharing and adopting trained models become very popular. On the other hand, this gives attackers new opportunities. Backdoor attack is one of the neural network attacks that often cause significant threats to the system. Backdoor is a hidden pattern trained into a neural network, which produces unexpected behavior if and only if a specific trigger is added to an input~\cite{NC}. In classification tasks, the backdoor misclassifies an arbitrary inputs to the same target label, when the associated trigger is applied to the input. Hence another desirable property of a neural network would be 'backdoor free', where the backdoor attack success rate is kept below a certain threshold. 

%\todo{Define the property} 
\begin{definition}[Backdoor Attack Success Rate]
\label{def:sr}
Let $N$ be a backdoored neural network, $t$ be the target label, $X$ be a set of adversarial inputs with the trigger and $N(x)$ be the prediction on $x$. We say attack success rate (SR) is:
\begin{equation}
\begin{aligned}
SR(t) = P(N(x) = t | x\in X)
\end{aligned}
\end{equation}
\end{definition}
Intuitively, attack success rate is the percentage of adversarial inputs classified into the target label. Next, we define backdoor-free neural network.

\begin{definition}[Backdoor-free Neural Network]
\label{def:backdoor}
Let $N$ be a backdoored neural network and $\xi$ be a positive real-value constant. $N$ satisfies backdoor-free property with respect to $\xi$ and $t$, if $SR(t) < \xi$.
\end{definition}

\begin{comment}
\begin{example}
We train a convolutional neural network $N_5$ on MNIST~\cite{mnist} to recognize 10 hand-written digits in gray-scale images. We apply backdoor attack proposed in~\cite{badnet} and choose '3' as the target label, so that the SR is $98.68\%$. The trigger is a $4\times4$ square located at the bottom right of the original image. In this example, we set $\xi$ to be $1\%$, i.e., the backdoor is effectively mitigated if the SR of the repaired network is $< 1\%$.
\end{example}
\end{comment}

\noindent \emph{Safety}: For neural networks applied on safety-critical systems such as the Airborne Collision Avoidance~\cite{acas} system, all safety properties must be satisfied strictly. However, due to reasons like limited training data and insufficient training, those critical properties could be violated. Hence, another desirable behavior of a neural network is to satisfy all safety properties or the violation rate is kept below a certain threshold. 

%\todo{Define the property}
\begin{definition}[Safety Property Violation Rate]
\label{def:vr}
Let $N$ be a neural network, $X$ be a set of inputs and $N(x)$ be the prediction on $x$. Let $\rho$ be the critical safety property that $N$ is expected to satisfy. We say the property violation rate (VR) is:
\begin{equation}
\begin{aligned}
VR(\rho) = P(N(x) \not \vDash \rho~|~x\in X)
\end{aligned}
\end{equation}
\end{definition}
Intuitively, violation success rate is the percentage of inputs that violate the property $\rho$. Next, we define safety property violation-free neural network.

\begin{definition}[Safety Property Violation-free Neural Network]
\label{def:safety}
Let $N$ be a neural network and $\xi$ be a positive real-value constant. $N$ is safety property violation-free with respect to $\xi$ and $\rho$, if $~VR(\rho) < \xi$.
\end{definition}

\begin{comment}
\begin{example}
As discovered by~\cite{Reluplex}, $N_{2,9}$ of ACAX Xu~\cite{acas} networks violates safety property $\phi_8$. We set $\xi$ to $1\%$ so that we aim to repair $N_{2,9}$ so that $VR(\phi_8)$ of the repaired network is kept below $1\%$. 
\end{example}
\end{comment}

\begin{example}
We train a feed-forward neural network $N$ on Census Income dataset (refer to details of the dataset in Section~\ref{sec:dataset}) to predict whether an individual's income exceeds \$50K per year. The neural network is configured to have 5 hidden layers. We use this network as the running example in this paper.
%with 64, 32, 16, 8 and 4 hidden neurons respectively. 
In this example, we focus on a fairness property $\phi$ w.r.t. protected feature \emph{gender}, i.e., the model is considered unfair if the probability difference of a favourable prediction for females and males is greater than $1\%$. Note that this is for illustration purpose only and such a threshold is probably too strict for real-world applications. Given this neural network, the problem is to find a repaired network $M$ with minimal adjustment to $N$'s weight parameters so that $M$ satisfies the fairness property whilst maintaining its accuracy.
\end{example}

\subsection{Our Problem} 
We are now ready to define our problem. 
\begin{definition}[Neural Network Repair Problem]
\label{def:problem}

Let $N$ be a neural network, $S$ be the input space, and $\phi$ be one of the above-mentioned desirable properties. The neural network repair problem is to construct a neural network $M$ such that $M$ satisfies $\phi$ in the input space $S$ and the semantic distance between $N$ and $M$ is minimized. 
\end{definition}

We define semantic distance between two neural networks $N$ and $M$ as follows.

\begin{definition}[Semantic Distance]
\label{def:semdis}
Let $N(x)$ be the prediction of neural network $N$ on input $x$ where $x \in S$, the semantic distance between $N~and~M$ is defined as:

\begin{equation}
\begin{aligned}
P(N(x) \neq M(x))
\end{aligned}
\end{equation}

\end{definition}
Intuitively, the semantic distance between two neural networks is the probability that the predictions of the two models are different on the same input. In this work, we use the model accuracy difference of $N$ and $M$ as a measure of their semantic distance.

 \section{Our Approach} \label{sec:approach}
In this section, we present the details of our approach. An overview of our framework is as shown in Figure \ref{fig:framework} and Algorithm~\ref{alg:CARE} shows the details of our approach. The first step is to verify whether the given neural network $N$ satisfy property $\phi$. If $\phi$ is satisfied, CARE terminates immediately. Otherwise, CARE proceeds to the next step. It performs causality-based fault localization on all hidden neurons of $N$ and identifies those neurons that have the most contribution to $N$'s unwanted behavior. The third step is to optimize these neurons' weights and generate a repaired network $M$ such that $M$ satisfies the property $\phi$ and is semantically close to $N$. 
%With given neural network $N$, input space $S$ and the desirable property $\phi$ as the input, CARE performs property verification, fault localization and repair step by step and finally generates the repaired network $M$.

\begin{figure}[t]
%\begin{center}
%\centering
%\resizebox{0.35\textwidth}{!}
\begin{tikzpicture}[node distance=1.2cm]
\node (start) [startstop] {Start};
%\node (in1) [io, below of=start, yshift=0.3cm] {JSON Input};
%\node (pro1) [process, below of=in1, yshift=0.3cm] {Model Parser};

\node (pro3) [process, below of=start] {Verification};

\node (dec1) [decision, right of=pro3, xshift=1.5cm, yshift=-1.2cm] {$\phi$ satisfied?};
\node (pro4a) [process_n, right of=dec1, xshift=1.5cm] {Fault Localization};

\node (out1) [io, right of=start, xshift=1.5cm, yshift=-0.9cm] {Output Results};
\node (pro5) [process_n, right of=out1, xshift=1.5cm] {Repair};
\node (stop) [startstop, right of=start, xshift=1.5cm] {Stop};

\draw [arrow] (start) -- (pro3);
%\draw [arrow] (in1) -- (pro1);
%\draw [arrow] (pro1) -- (pro3);
\draw [arrow] (pro3) |- (dec1);

\draw [arrow] (dec1) -- node[anchor=west] {Yes} (out1);
\draw [arrow] (dec1) -- node[anchor=south] {No} (pro4a);

\draw [arrow] (out1) -- (stop);
\draw [arrow] (pro4a) -- (pro5);
\draw [arrow] (pro5) -- (out1);

\end{tikzpicture}
            \centering 
\caption{An overview of our framework }
\label{fig:framework}
%\end{center}
\end{figure}
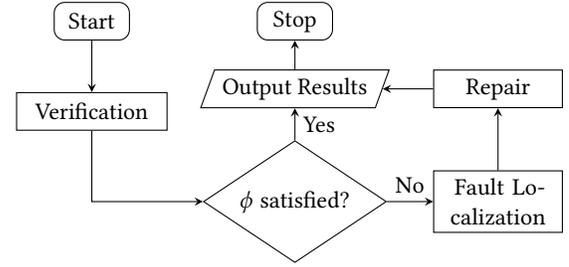

\begin{algorithm}[t]
\DontPrintSemicolon
    Verify property $\phi$ on $N$; \\
    Perform accuracy check on $N$;    \\
    \If{$\phi$ is verified}
    {
        \Return{``Verified'';}
    }
    \For{all hidden neuron $x$ in $N$}{
        Calculate $ACE_{do(x)}^{y}$;
    }
    
    Sort $ACE$ for all $x$;\\
    Candidate neurons $\leftarrow top~10\%$; \\
    
    \Do{$\#iteration \le 100$}
    {
        $Candidate~M \leftarrow$ PSO~searching;\\
        Verify property $\phi$ on $Candidate~M$; \\
        \If{$\phi$ is verified}
        {
            break;
        }
        \If{Failed to find a better location in last 10 consecutive search}
        {
            break;
        }
    }
    $M=Candidate~M$;\\
    Perform accuracy check on $M$;\\
    \If{$Acc(N)~-~Acc(M) > threshold$}
    {
        \Return{``Significant accuracy drop!"};
    }
    Verify property $\phi$ on M;\\
    \If{$\phi$ is verified}
    {
        \Return{$M$;}
       
    }
    \Else{
      \Return{``Not able to repair the network!"};
    }

\caption{$CARE(N, S,\phi)$}
\label{alg:CARE}
\end{algorithm}

\subsection{Property Verification}
In this step, we verify neural network $N$ against given property $\phi$. If $N$ satisfies $\phi$, CARE returns $N$ unchanged and terminates. %Otherwise, CARE verifies $Ms$ model accuracy as a benchmark to measure the performance drop in repaired network $M$. 
Otherwise, CARE proceeds to repair the neural network. Property verification is not the focus of this work and we adopt existing neural network verification techniques. Recently there have been multiple tools and approaches for solving the above-mentioned neural network verification problems. We omit the details on how different properties are verified and refer the readers to~\cite{FM21,verification,marabou20,Reluplex}. In our implementation, CARE's fairness verification algorithm is based on~\cite{FM21}, backdoor success rate verification is based on the method proposed in~\cite{NC} and safety property verification is based on the approach proposed in~\cite{provable}.

%In this step, existing fairness verification techniques and tools that support the architecture of $N$ can be applied, such as AEQUITAS~\cite{aequitas}, AIF360~\cite{aif360}, and VeriFair~\cite{verification}. The verification method proposed in~\cite{aequitas} is to sample a large set of instances and measure the percentage of discriminatory instances. As machine learning model inference is often optimized nowadays, this method can be very efficient. AIF360 is a framework that implements multiple algorithms to detect bias in machine learning models. It is an open-source toolkit with Python library available for easy integration. VeriFair verifies machine learning model fairness via concentration and is proved to be effective and efficient~\cite{verification}. Since it is not the focus of this work, we omit details on how fairness verification of neural networks work. In our implementation, CARE's fairness verification algorithm is based on the sampling approach proposed in~\cite{DBLP:conf/cav/BazilleGJS20}. Note that Algorithm~\ref{alg:CARE} additionally checks whether the model accuracy is satisfactory, which is a normal practice. 
%\todo{We should add a few sentence on how effective these existing verification techniques. The idea is to make it clear that these are quite scalable and thus this step won't be a bottleneck of our approach.}

\begin{example}
In our running example, CARE verifies $N$ against the fairness property $\phi$. The resultant fairness score (i.e., the probability difference of a favourable prediction for females and males) is $1.30\%$. Therefore, $N$ fails to satisfy the fairness property (which requires the fairness score to be within $1\%$). 
%\todo{Add the verification result of the other two models, one for backdoor and one for safety.}
\end{example}

\subsection{Causality Analysis}
In this step we perform causality-based fault localization. The goal is to identify a relatively small set of neurons to repair. Note that in this work, we restrict ourselves to repair by adjusting the parameters of existing neurons only, i.e., without introducing or pruning neurons. The number of parameters in a neural network is often huge and even relatively small neural network may consist of thousands of weight and bias parameters. Therefore, attempting to optimize all neural weights would be costly and hard to scale. CARE thus adopts a fault localization method to identify neurons that are likely to be responsible for the undesirable behaviors. This method is based on causality analysis carried out over all hidden neurons of $N$. 

An alternative method for fault localization is gradient-based method, which is to measure how much perturbing a particular neuron would affect the output~\cite{arachne}. CARE is designed to capture the causal influence of a hidden neuron on the performance on the given property. In contrast, gradient-based method draws conclusion based on statistical correlations, which is not ideal for model repair (since our aim is to identify neurons that \emph{cause} the defect). Next, we describe our causality-based fault localization in detail.

In recent years, causality has gained increasing attention in interpreting machine learning models~\cite{cal_att,cal2}. Multiple approaches have been designed to explain the importance of the components in a machine learning model when making a decision, based on causal attributions. Compared with traditional methods, causal approaches identify causes and effects of a model's components and thus facilitates reasoning over its decisions.

In the following, we review some concepts which are necessary for the causality analysis in this work. 
%\todo{Why SCM not others? We need to say something like SCM is most important/popular model for causality analysis.}

\begin{definition}[Structural Causal Models (SCM)~\cite{pearl_causality}]
\label{def:SCM}
A Structural Causal Model (SCM) is a 4-tuple $M(X,U,f,P_u)$ where X is a finite set of endogenous variables, U denotes a finite set of exogenous variables, f is a set of functions $\{f_1, f_2,...,f_n\}$ where each function represents a causal mechanism such that $\forall{x_i \in X}, x_i=f_i(Pa(x_i),u_i)$ where $Pa(x_i)$ is a subset of $~X\setminus\{x_i\}$, $u_i \in U$ and $P_u$ is a probability distribution over $U$.
\end{definition}

SCM serves as a framework for representing what we know about the world, articulating what we want to know and connecting the two together in a solid semantics~\cite{7tools}. It plays an important role in causality analysis and is commonly applied in many studies~\cite{cal_att,causal_dnn,cal_visual,cal_bb}. 

Figure~\ref{fig:scm} shows an example causal graph to study the efficiency of a medication on a disease, where the nodes represent variables and the edges represent cause-effect relations intuitively. In this graph, age is considered as an exogenous variable (i.e., confounder), patient's heart rate, cholesterol level and whether the medication is applied or not are endogenous variables (i.e., whether the medication is applied or not is often considered as the treatment). The outcome is the recovery rate of a patient. As illustrated in the graph, age affects the patient's health conditions such as heart rate and level of cholesterol. Furthermore, the need or feasibility of applying this medicine on patients is affected by age, i.e., young people may not necessarily take the medicine and patients above 70 years old are too risky to take the medicine.  Patient's health condition and the application of medication affect the recovery rate. Furthermore, age can affect the recovery rate directly since younger patient often recover faster than the elderly.

%\todo{The example here is too abstract; can we have some intuitive example instead?}

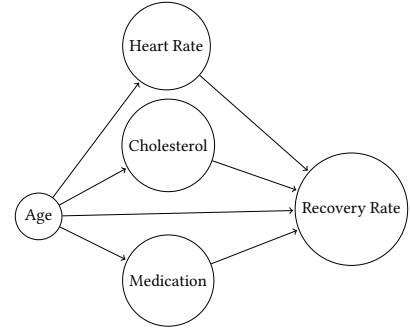
\begin{figure}
    \resizebox{0.3\textwidth}{!}{
       \begin{tikzpicture}
            % Add the notes
            \node[state]             (age) {Age};
            \node[state, above right=of age,  yshift=-0.6cm, xshift=0.5cm]  (cho) {Cholesterol}; 
            \node[state, above right=of age,  yshift=1.3cm, xshift=0.5cm] (hr) {Heart Rate};
            \node[state, below=of cho,  yshift=0.2cm] (med) {Medication};
            \node[state, right=of cho, yshift=-1.2cm, xshift=0.5cm] (rec) {Recovery Rate};

            % Connect the states with arrows
            \draw[every loop]
                (age) edge[left, auto=left]  node {} (hr)
                (age) edge[left, auto=left]  node {} (cho)
                (age) edge[left, auto=left]  node {} (med)
                (age) edge[left, auto=left]  node {} (rec)
                (hr) edge[left, auto=left]  node {} (rec)
                (cho) edge[left, auto=left]  node {} (rec)
                (med) edge[left, auto=left]  node {} (rec);
           
        \end{tikzpicture}
    }
\caption{An Example Causal Graph}
\label{fig:scm}
\end{figure}

In this work, we model neural networks as SCMs to analyze the causal relationship between hidden neurons and model's predictions. To debug and repair a neural network, we would like to measure the `contribution' of each neuron to the network misbehavior, which is referred as the causal effect. 

\begin{definition}[Average Causal Effect]
\label{def:ACE}
The Average Causal Effect (ACE) of a binary random variable $x$ (e.g., treatment), on another random variable $y$ (e.g., outcome) is defined as follows.
\begin{equation}
\label{eq:ace}
ACE = \mathbb{E}[y|do(x=1)]-\mathbb{E}[y|do(x=0)] 
\end{equation}
where $do(\cdot)$ operator denotes the corresponding interventional distribution defined by SCM.
\end{definition}
Intuitively, ACE is used to measure the causal effect of $x$ on $y$ by performing intervention on $x$. There are many other causal effect metrics in the literature such as 
\begin{itemize}
    \item average treatment effect (ATE), i.e., $ATE = \mathbb{E}[y(1)-y(0)]$, where $y(1)$ represents the potential outcome if a treatment is applied exogenously and $y(0)$ represents the corresponding potential outcome without treatment. ATE measures the effect of the treatment at the whole population level~\cite{ate},
    \item conditional average treatment effect, i.e., $CATE=\mathbb{E}[y(1)-y(0)|x_0]$~\cite{cate}, which is the average treatment effect conditioned on a particular subgroup of $X$, i.e., $X=x_0$, 
    %\todo{seems missing some intuitive explanation of this one here}
    \item and effect of treatment on the treated, i.e., $ETT=\mathbb{E}[y_{x_1}|x_0]-E[y|x_0]$ which measures the probability of $Y$ would be $y$ had $X$ been $x_1$ counterfactually, given that in the actual world $X=x_0$~\cite{causal_learning}.
\end{itemize} 
In this work, we focus on the ACE metric since we are interested in measuring the causal effect of the hidden neurons on certain output value.

\subsection{Fault Localization}
In the following, we present details on how causality analysis is used for fault localization in neural networks. 
Firstly, the neural network $N$ is modeled as an SCM. As proposed in~\cite{NNSCM}, neural networks can be interpreted as SCMs systematically. In particular, feed-forward neural networks (FFNNs) and convolutional neural networks(CNNs) can be represented as directed acyclic graphs with edges from an earlier (i.e., closer to the input layer) layer to the next layer until the output layer. The following is a proposition. %from~\cite{cal_att}.

\begin{proposition}
\label{prop:scm}
An $n$-layer FFNN or CNN $N(x_1, x_2,...,x_n)$ where $x_i$ represents the set of neurons at layer $i$, can be interpreted by SCM $M([x_1, x_2,...,x_n], U, [f_1, f_2,...,f_n], P_U)$, where $x_1$ represents neurons at input layer and $x_n$ represents neurons at output layer. Corresponding to every $x_i$, $f_i$ represents the set of causal functions for neurons at layer $i$. $U$ represents a set of exogenous random variables that act as causal factors for input neurons $x_1$ and $P_u$ is a probability distribution over $U$.
\end{proposition}

\begin{proof}

In the scenario of FFNN, the proof of Proposition~\ref{prop:scm} follows that provided in~\cite{cal_att}. In the scenario of CNN, similar to FFNN, neurons at each layer can be written as functions of neurons at its previous layer, i.e., $\forall{i}:~\forall{x_{i_j}} \in x_i:~x_{i_j} = f_{i_j}(x_{i - 1})$, where $x_{i_j}$ represents the $j^{th}$ neuron at layer $i$. Neurons at input layer ($x_1$) can be assumed to be functions of independent noise variables $U$ such that $\forall{x_{1_j} \in x_1}~and~u_j \in U:~x_{1_j} =f_{1_j}(u_j)$. Thus a CNN can be equivalently expressed by a SCM $M([x_1, x_2,...,x_n], U, [f_1, f_2,...,f_n], P_U)$.

\end{proof}

Figure \ref{fig:ffnn_scm} depicts the SCM of a 3-layer FFNN. The dotted circles represent exogenous random variables which act as causal factors for the input neurons. In this work, we assume that neurons at the input layer are not causally related to each other but can be jointly caused by a latent confounder.

%\todo{CNN modelled as SCM here.}

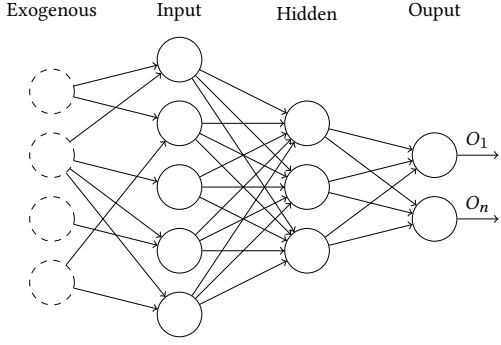
\begin{figure}
    \tikzset{%
    every neuron/.style={
        circle,
        draw,
        minimum size=0.7cm
      }
    }
    
    \tikzset{% \draw [dashed, line width=1pt]circle(0.655cm);
    every exo/.style={
        circle,
        draw,
        dashed,
        minimum size=0.7cm
      }
    }
\resizebox{0.38\textwidth}{!}{
    \begin{tikzpicture}%[x=1.5cm, y=1.5cm, >=stealth]
        \pgfplotsset{%
        width=0.3\textwidth,
        height=0.2\textwidth
        }
        
        \foreach \m/\l [count=\y] in {1,2,3,4}
          \node [every exo/.try, exo \m/.try] (exogenous-\m) at (0,2-\y) {};
        
        \foreach \m [count=\y] in {1,2,3,4,5}
          \node [every neuron/.try, neuron \m/.try ] (input-\m) at (2,2.5-\y) {};
          
        \foreach \m [count=\y] in {1,2,3}
          \node [every neuron/.try, neuron \m/.try ] (hidden-\m) at (4,1.5-\y) {};
        
        \foreach \m [count=\y] in {1,2}
          \node [every neuron/.try, neuron \m/.try ] (output-\m) at (6,1-\y) {};
        
        %\foreach \l [count=\i] in {1,2,3,n}
        %  \draw [<-] (input-\i) -- ++(-1,0)
        %    node [above, midway] {$I_\l$};
        
        %\foreach \l [count=\i] in {1,n}
        %  \node [above] at (hidden-\i.north) {$H_\l$};
        
        \foreach \l [count=\i] in {1,n}
          \draw [->] (output-\i) -- ++(1,0)
            node [above, midway] {$O_\l$};
        
        %\foreach \i in {1,...,4}
        %  \foreach \j in {1,...,5}
        %   \draw [->] (exogenous-\i) -- (input-\j);
        \draw [->] (exogenous-1) -- (input-1);
        \draw [->] (exogenous-1) -- (input-2);
        \draw [->] (exogenous-2) -- (input-1);
        \draw [->] (exogenous-2) -- (input-3);
        \draw [->] (exogenous-2) -- (input-4);
        \draw [->] (exogenous-2) -- (input-5);
        \draw [->] (exogenous-3) -- (input-4);
        \draw [->] (exogenous-4) -- (input-2);
        \draw [->] (exogenous-4) -- (input-5);
            
        \foreach \i in {1,...,5}
          \foreach \j in {1,...,3}
            \draw [->] (input-\i) -- (hidden-\j);
        
        \foreach \i in {1,...,3}
          \foreach \j in {1,...,2}
            \draw [->] (hidden-\i) -- (output-\j);
        
        \foreach \l [count=\x from 0] in {Exogenous, Input, Hidden, Ouput}
          \node [align=center, above] at (\x*2,2) {\l };
    
    \end{tikzpicture}
}
\caption{FFNN as an SCM}
\label{fig:ffnn_scm}
\end{figure}

Next we define the attribution problem, i.e., what is the causal influence of a particular hidden neuron on model defect. Recall that in Definition~\ref{def:ACE}, we define the ACE of a binary variable $x$ on output variable $y$. However, we cannot apply the definition directly for two reasons. First, the domain of neural networks' neurons is mostly continuous, not binary-valued. Second, we are interested in finding causal effects on the model defect rather than certain output variable. Hence we propose the following definition.

\begin{definition}[Causal Attribution]
We denote the measure of the undesirable behavior of given neural network $N$ as y. The Causal Attribution of a hidden neuron $x$ to $N's$ defect $y$ is:

\begin{equation}
\label{eq:attri}
ACE_{do(x=\beta)}^{y} = \mathbb{E}[y|do(x=\beta)]
\end{equation}
\end{definition}

%We refer the readers to~\cite{cal_att} for the proof of Proposition \ref{prop:scm} and the reasoning of selecting average ACE of $x_i$ as $baseline_{x_i}$. 

Next, we calculate the interventional expectation of $y$ given intervention $do(x=\beta)$ and we thus have the following.
\begin{equation}
\label{eqn:e_ace}
\mathbb{E}[y|do(x=\beta)]=\int_y yp(y|do(x=\beta))dy
\end{equation}
%\todo{We need to add a sentence to explain the intuition of the definition} 
Intuitively, causal attribution measures the effect of neuron $x$ being value $\beta$ on $y$. We evaluate Equation~\ref{eqn:e_ace} by sampling inputs according to their distribution whilst keeping the hidden neuron $x = \beta$, and computing the average model undesirable behavior $y$. 

%\todo{Formalize the following, using mathematical notations. Right now, it is too informal and not clear what is the input and what is the output. }

In this work, we calculate $y$ according to the desirable property $\phi$. Let $N_t(x_{ip})$ be the prediction value of class $t$ on input $x_{ip}$ we have:
\begin{itemize}
    \item For fairness repair, we measure model unfairness $y$ by taking the difference of the prediction value on favourable class $t$ w.r.t. samples that only differ by the sensitive feature. Let $x_{ip}$ and $x_{ip}'$ be a pair of discriminatory instances~\cite{adf2020} (that only differ by the sensitive feature). We have $y_{fair} = |~N_t(x_{ip}) - N_t(x_{ip}')~|$ and we calculate $ACE_{do(x=\beta)}^{y_{fair}}$ as the causal attribution of neuron $x$ w.r.t. fairness property.
    \item For backdoor removal, we measure $y$ by calculating the prediction value on target class $t$, i.e, $y_{bd} = N_t(x_{ip})$. Thus we aim to calculate $ACE_{do(x=\beta)}^{y_{bd}}$ for all hidden neurons. Note $x_{ip}$ can be clean inputs or adversarial inputs.
    \item For safety property violation repair, we measure $y$ by calculating the prediction value on desirable class $t$, i.e., $y_{safe} = N_t(x_{ip})$. We calculate $ACE_{do(x=\beta)}^{y_{safe}}$ as the causal attribution for hidden neuron $x$.
\end{itemize} 
 For other properties, CARE is able to calculate the causal attribution of each neuron on $y$ as long as the corresponding $y$ is specified. 

Thus by calculating the causal attribution of each hidden neuron on $y$, we are able to identify neurons that are most likely the cause to the unexpected behavior. Afterwards, the identified neurons are used as candidates for model repair in our next step. The time spent on this step depends on the total number of neurons $n$ to analyze in the given neural network. The time complexity for causality-based fault localization is thus $\textbf{O}(n)$.

%We discuss and evaluate the time complexity of this step in RQ3.
%\todo{If space allows, discuss the complexity of the approach here.}

\begin{example}
In our running example, CARE conducts causality-based fault localization on all hidden neurons of $N$ to identify neurons that are the most ``responsible" to $N's$ prediction bias. CARE generates 12.4k samples to analyze all neurons and the total time taken is 120s.
%\todo{Present some details, like how many samples etc.}
%Figure~\ref{fig:example} shows the scatter plot for the average causal attribution of each hidden neuron, where $L\_N$ represents $N^{th}$ neuron at layer $L$. 
It is observed that a few neurons (such as the $11^{th}$ neuron at $3^{rd}$ hidden layer and the $3^{rd}$ neuron at $5^{th}$ hidden layer) have a contribution which is ``outstanding". This is good news to us as it implies that it is possible to improve the fairness by modifying the weights of only a few neurons. CARE then selects the top $10\%$ of total number of neurons (i.e., 13 in total for this network) to be used as repair candidates.
%\todo{If space allows, we need to some the other two examples as well.}
\end{example}

%\begin{comment}
 
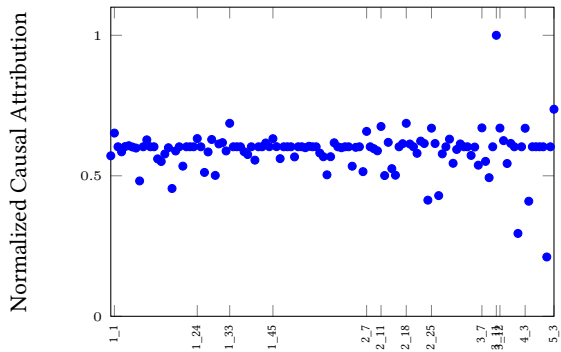
\begin{figure}
\begin{tikzpicture}
\pgfplotsset{%
    width=0.42\textwidth,
    height=0.32\textwidth
}
\pgfplotsset{every tick label/.append style={font=\tiny}}
\begin{axis}[
    ylabel={Normalized Causal Attribution},
    xmin=0, xmax=123,
    ymin=0.0, ymax=1.1,
    xtick = {1,24,33,45,71,75,82,89,103,107,108,115,123},
    xticklabels={1\_1,1\_24,1\_33,1\_45,2\_7,2\_11,2\_18,2\_25,3\_7,3\_11,3\_12,4\_3,5\_3},
    x tick label style={rotate=90,anchor=east}],
    ytick={0,0.2, 0.4, 0.6, 0.8, 1},
    %legend pos=north west,
    ymajorgrids=false,
    grid style=dashed,
    legend style={
    legend pos= north east,
    }
]

\addplot[
    only marks,
    mark size=1.5pt,
    %scatter,    
    %mark=square,
    color = blue,
    ]
    coordinates {
    (0,0.571140990829037)(1,0.651976568069838)(2,0.603267365842595)(3,0.585329419444066)(4,0.604115342368248)(5,0.607080730164652)(6,0.602490521296457)(7,0.598738842914739)(8,0.481721420024747)(9,0.603267365842595)(10,0.627973056708016)(11,0.603267365842595)(12,0.603267365842595)(13,0.560304445078594)(14,0.550911141990349)(15,0.577648452505017)(16,0.599957272078553)(17,0.454800242920928)(18,0.58829399692471)(19,0.603267365842595)(20,0.534325645705698)(21,0.603267365842595)(22,0.603267365842595)(23,0.603267365842595)(24,0.632556828755716)(25,0.603267365842595)(26,0.511989073452484)(27,0.584926733428382)(28,0.629221581314646)(29,0.50141778386179)(30,0.613395934500678)(31,0.618008713872888)(32,0.588469894097546)(33,0.68673086587589)(34,0.603267365842595)(35,0.603267365842595)(36,0.603267365842595)(37,0.585815246771789)(38,0.575646229640254)(39,0.603267365842595)(40,0.555670755578488)(41,0.603267365842595)(42,0.603267365842595)(43,0.616754680858597)(44,0.603577511094338)(45,0.63205937767204)(46,0.603717820707171)(47,0.561106987710295)(48,0.603267365842595)(49,0.603267365842595)(50,0.603267365842595)(51,0.567270369869375)(52,0.603267365842595)(53,0.603267365842595)(54,0.599967268079312)(55,0.605106325779989)(56,0.603267365842595)(57,0.603267365842595)(58,0.581038371454958)(59,0.567873800035709)(60,0.503473083760344)(61,0.56794723815664)(62,0.617546323152091)(63,0.603267365842595)(64,0.599871951218455)(65,0.603267365842595)(66,0.603267365842595)(67,0.534149853117362)(68,0.601303376457428)(69,0.603267365842595)(70,0.514801975271146)(71,0.658056409181758)(72,0.603267365842595)(73,0.596651632221142)(74,0.588915560416809)(75,0.675541182328619)(76,0.500901880989274)(77,0.61906609102851)(78,0.525500401197601)(79,0.502086182698758)(80,0.603267365842595)(81,0.614232479950806)(82,0.686756186309212)(83,0.613148168343575)(84,0.603267365842595)(85,0.580044150145765)(86,0.623474864306573)(87,0.614931876586814)(88,0.41332545006709)(89,0.669454049356509)(90,0.614950660890937)(91,0.429480033256999)(92,0.578056147504263)(93,0.603267365842595)(94,0.630422804027546)(95,0.544523907372246)(96,0.593707180956621)(97,0.613409764399897)(98,0.603267365842595)(99,0.603267365842595)(100,0.572508333705071)(101,0.602299543043309)(102,0.537769111949337)(103,0.670557876563125)(104,0.5513347749449)(105,0.493294497540777)(106,0.603267365842595)(107,1)(108,0.669703731510302)(109,0.625176347889641)(110,0.543768932224302)(111,0.615304386943625)(112,0.603267365842595)(113,0.29489135932914)(114,0.603267365842595)(115,0.669187710202817)(116,0.409508315980534)(117,0.603267365842595)(118,0.603267365842595)(119,0.603267365842595)(120,0.603267365842595)(121,0.211110719864358)(122,0.603267365842595)(123,0.736880505409135)
    };
    %\addlegendentry{Hidden Neurons}    
    
\end{axis}
\end{tikzpicture}
%\centering
%\begin{minipage}{0.47\textwidth} % choose width suitably
%\includegraphics[width=\linewidth]{mygraphicfile}
%{Fx denoting feature x. Nx\_y represents neuron y at hidden layer x.\par} %In this experiment $\epsilon = 0.01, \delta = 0.001$.
%\end{minipage}
\caption{Causal Attribution of Neurons in the Example}
\label{fig:example}
\end{figure}
%\end{comment}

\subsection{Network Repair}
\label{sec:repair}
Next, we present how CARE performs neural network repair. Similar to traditional software programs, neural networks can go wrong and incorrect behaviors should be fixed. While existing efforts mainly focus on how to retrain the neural network to repair the unexpected behaviors, there are a few recent works~\cite{nnrepair,minimal20,digits} that address this problem by modifying the parameters of the network. In the following, we first briefly explain why these approaches do not apply in our setting. Goldberger \emph{et al.} propose a verification-based method to adjust the neural weights of the output layer only. Their approach has limited effectiveness in repairing neural networks as demonstrated in~\cite{gl_rep}. Their heuristics (i.e., focusing on neurons in the output layer) is not justified from our point of view as the neurons causing the problematic behaviors may not be in the output layer. In~\cite{digits}, an approach is proposed to repair a model through inductive synthesis. Their approach has scalability issues as a large model size leads to a large search space. While NNRepair~\cite{nnrepair} provides a constraint-based approach to fix the logic of a neural network at an intermediate layer or output layer, its fault localization and repair applies to a single layer only but the actual fault could be across multiple layers. Thus its repair performance may not be ideal and we provide a comparison in 
Section~\ref{subsec:rq}.

In our approach, we propose to address the repair problem through optimization, i.e., based on the Particle Swarm Optimisation (PSO) algorithm. 
%\todo{We need to justify why PSO is the choice and not other typical optimization techniques such as GA, MCMC and so on.} 
Besides PSO, there are many other optimization algorithms. Genetic optimization algorithms (genetic algorithm for example) usually do not handle complexity in an efficient way~\cite{ga_pso} since the large number of elements undergoing mutation causes a considerable increase in the search space. In comparison, PSO requires smaller number of parameters and thus has lower number of iterations. 
%GA often converges towards a local optimum or arbitrary points but PSO tries to find the global optima~\cite{ga_pso}. 
Another type of optimization is stochastic optimization algorithm. Markov Chain Monte Carlo (MCMC) as an representative uses a sampling technique for global optimization. However, MCMC often encounter either slow convergence or biased sampling~\cite{mcmc}. Hence, in this work we select intelligence based algorithm PSO. 

The idea is to search for small adjustments in weight parameters of the neurons identified in the previous step such that the specified property is satisfied. PSO simulates intelligent collective behavior of animals such as schools of fish and flocks of birds. It is known to be particularly effective for optimization in continuous domains~\cite{pso}. In PSO, multiple particles are placed in the search space. At each time step, each particle updates its location $\overrightarrow{x_i}$ and velocity $\overrightarrow{v_i}$ according to an objective function. That is, the velocity is updated based on the current velocity $\overrightarrow{v_i}$, the previous best location found locally $\overrightarrow{p_i}$ and the previous best location found globally $\overrightarrow{p_g}$. Its location is updated based on the current location and velocity. We write $R(0,c)$ to denote a random number uniformly sampled from the range of $[0, c]$. Formally, the PSO update equation is as follows~\cite{pso1998}.
\begin{align}
    \overrightarrow{v_i} & \leftarrow \omega\overrightarrow{v_i} + R(0, c_1)(\overrightarrow{p_i} - \overrightarrow{x_i}) + R(0, c_2)(\overrightarrow{p_g} - \overrightarrow{x_i}) \\
    \overrightarrow{x_i} & \leftarrow \overrightarrow{x_i} + \overrightarrow{v_i}
\end{align}
where $\omega$, $c_1$ and $c_2$ represent inertia weight, cognitive parameter and social parameter respectively. In PSO, the fitness function is used to determine the best location.

In CARE, the weights of the identified neurons are the subject for optimization and thus are represented by the location of the particles in PSO. The initial location of each particle is set to their original weight and the initial velocity is set to zero.

As defined in Section \ref{sec:preli}, our problem is to repair a given neural network $N$ against property $\phi$ while minimizing the accuracy drop. Therefore, two essential components need to be considered in the optimization process: 1) model performance on property $\phi$ and 2) model accuracy. In CARE, we measure the model performance based on the property $\phi$ specified. Formally, the fitness function of PSO is defined as follows.
\begin{equation}
    fitness = (1 - \alpha)~UB + \alpha (1 - accuracy)
\end{equation}
where UB (undesirable behavior) is a quantitative measure on the degree of violating property $\phi$; constant parameter $\alpha \in (0,1)$ determines the relative importance of the accuracy. For fairness improvement task, we randomly sample a large set of instances and $UB$ is measured by the percentage of individual discriminatory instances within the set, i.e., let $N(x)$ be the prediction of $x$ and $(x, x')$ be a pair of discriminatory instances in the sample set, $UB = P(N(x) \neq N(x'))$; for backdoor removal task, we add backdoor trigger (if it is known) to all samples in testing set and measure $UB = SR(t)$ following Definition~\ref{def:sr}; for safety task, we randomly sample a large set of instances and measure $UB = VR(\rho)$ following Definition~\ref{def:vr}. Intuitively, our optimization target is to make the repaired model satisfy given property $\phi$ while maintaining the accuracy. 
%\todo{Update the text to measure ``sentence''; We may have to discuss a bit on how many samples we use and why. If it is hard to explain, we can leave it as an empirical parameter.}
%For fairness repair tasks, we use the percentage of individual discriminatory instances~\cite{adf} in the sampled set as an indicator of model unfairness. Therefore, an increase in the percentage of discriminatory instances incur penalty to our optimization. For backdoor removal task, we use the attack success rate (the percentage of adversarial inputs classified into the target label) as an indicator. For safety property violation repair task, we use the rate of property violation on sampled inputs as an indicator. 
Note that, for other properties, CARE is able to conduct repair as long as a performance indicator is specified based on the given property properly. 
%Furthermore, model accuracy is monitored in the process to make sure the semantic difference between original model $N$ and repaired model candidate $M$ is minimized. 

%\todo{The accuracy check part can be explained a bit more. It is currently not part of the overall workflow. What happens if accuracy is too low?}
\begin{example}
In our running example, CARE applies PSO to generate a repair network $M$ that will improve the fairness performance of $N$. The weights of the identified 13 neurons in the previous step are the subjects for the optimization. We set $\alpha$ to $0.8$ so that PSO can optimize for fairness without sacrificing model accuracy too much. PSO terminates at $17^{th}$ iteration where no better location is found in the last 10 consecutive iterations. By adjusting neural weights accordingly, $M$ is generated. CARE performs fairness and accuracy check on $M$ and the maximum probability difference is $0.007$ with accuracy of $0.86$ (original accuracy is $0.88$). Thus, CARE successfully repairs $N$ with a fairness improvement of $46.1\%$ and the model accuracy is mildly affected. 
\end{example}

%\todo{If space permits, we can add an algorithm to show the approach, including all the steps.}

\section{Implementation and Evaluation} \label{sec:imple}
%SOCRATES is a unified framework for developing and evaluating verification techniques for neural networks. Following the architecture of SOCRATES, user inputs are provided in a JSON format. The inputs include the neural netowork model to repair, the target fairness property as well as the required parameters such as parameter $\alpha$  used in the PSO fitness function. Our code is open-sourced and is available at~\cite{anoynoumous}. 
In the following, we conduct a set of experiments to evaluate CARE. We demonstrate the technique in the context of three different tasks: 1) neural network fairness improvement, 2) neural network backdoor removal and 3) neural network safety property violation repair. All experiments are conducted on a machine with 1 Dual-Core Intel Core i5 2.9GHz CPU and 8GB system memory. To reduce the effect of randomness, all experimental results are the average of five runs if randomness is involved. 

\subsection{Experiment Setup} \label{sec:dataset}
In the following, CARE is evaluated over multiple datasets by answering mutpile research questions (RQs). The details of the datasets are summarized as follows:

For fairness repair, CARE is evaluated over three datasets that are commonly used in machine learning model fairness testing~\cite{aequitas,SG,adf2020,counterfactual}, 

\begin{itemize}
    \item \emph{Census Income~\cite{census}:}  This dataset consists of 32,561 training instances containing 13 features and is used to predict whether an individual income exceeds \$50K per year. Among all attributes, gender, age and race are three protected characteristics. The labels are if an individual income exceeds \$50K per year or not.
    \item \emph{German Credit~\cite{credit}:}  This dataset consists of 1000 instances with 20 features and is used to assess an individual's credit. Here, age and gender are the protected features. The labels are whether a person's credit is good or not.
    \item \emph{Bank Marketing~\cite{bank}:} This dataset consists of 45,211 instances and there are 17 features. Among them, age is the protected feature. %The dataset aims at training models to predict whether the client will subscribe a term deposit. 
    The labels are whether the client will subscribe a term deposit or not.
    %\item \emph{Jigsaw Comment~\cite{jigsaw}:} This dataset consists of around 313,000 text comments with average length of 80 words classified into six categories of toxicity (i.e., toxic, severe toxic, obscene, threat, insult and identity hate) and non-toxic. The protected features analysed are race and religion.
\end{itemize}

We train three feed-forward neural networks
%and one LSTM RNN 
following standard practice and run CARE to repair the unfair behavior of each model against the corresponding protected features.

For backdoor removal, CARE is evaluated over three datasets:
\begin{itemize}
    \item \emph{German Traffic Sign Benchmark Dataset (GTSRB)~\cite{GTSRB}:} This dataset consists of 39.2K training instances and 12.6K testing instances of colored images. The task is to recognize 43 different traffic signs. We train a CNN network consists of 6 convolutional layers and 2 dense layers.
    \item \emph{MNIST~\cite{mnist}:}  This dataset consists of 70K instances of hand-written digits as gray-scale images. We train a standard 4-layer CNN network to recognize the 10 digits (0-9).
    \item \emph{Fashion-MNIST~\cite{fashion}:} This data set consists of 70K instances of 10 fashion categories, e.g., dress, coat and etc. Each sample is a 28x28 grayscale image. We train a CNN network consists of 3 convolutional layers and 2 dense layers.
\end{itemize}

%\todo{Just one network in this category seems few - are there more? In fact, it will sound better if we have 3 models for each property.}
For safety violation repairing, CARE is evaluated over three ACAS Xu~\cite{acas} networks. ACAS Xu contains an array of 45 DNNs that produces horizontal maneuver advisories of unmanned version of Airborne Collision Avoidance System X. As discovered in~\cite{Reluplex,Pham2020SOCRATESTA}, some DNNs violates certain safety properties, e.g., DNN $N_{2,9}$ violates the safety property $\phi{_8}$. We apply CARE on 3 sub networks $N_{2,9}$, $N_{3,3}$ and $N_{1,9}$ against 3 properties $\phi_8$, $\phi_2$ and $\phi_7$ respectively, aiming to repair the misbehavior.
Table~\ref{tab:dataset} shows the details of the trained networks used in our experiment.

%\todo{The numbers in the following paragraph are introduced without justification. I suggest we move the paragraph to the experiment section.}
In PSO, we follow the general recommendation in~\cite{psop} and set parameter $\omega = 0.8$, $c_1=c_2=0.41$ and number of particles to 20. 
%We set the bounds of weight adjustment to $(0,2)$, i.e., 0 to 2 times of the original weight. 
The maximum number of iteration is set to 100. To further reduce the searching time, we stop the search as soon as the property is satisfied or we fail to find a better location in the last 10 consecutive iterations. Note that CARE performs accuracy check after PSO finds an $M$. If the accuracy drop is significant, i.e., bigger than a threshold of $3\%$, CARE returns error message and a larger value of $\alpha$ shall be set in such scenarios.

\begin{table}[t]
\centering
\caption{Neutral Networks Used in Our Experiments}
\begin{tabular}[t]{p{0.5cm}p{1.4cm}p{2.9cm}p{0.9cm}p{1.0cm}}
\toprule
 Model & Dataset&  Architecture & \#Neuron&  Accuracy \\
\midrule
$NN_1$&Census&  6-layer FFNN& 139   & 0.8818\\
%NN_2&Census&  6-layer FFNN& 263   & 0.8160\\
%NN_3&Census&  6-layer FFNN& 511   & 0.8806\\
$NN_2$&Credit&    6-layer FFNN& 146   & 1.0\\
$NN_3$&Bank&    6-layer FFNN& 143   & 0.9226\\

$NN_4$&GTSRB&    6-Conv + 2 Dense CNN& 107,595   & 0.9657\\
$NN_5$&MNIST&    2-Conv + 2 Dense CNN& 31,610   & 0.9909\\
$NN_6$&Fashion&    3-Conv + 2 Dense CNN& 67,226   & 0.9136\\
$NN_7$&ACAS $N_{2,9}$&    6-layer FFNN& 305   & - \\
$NN_8$&ACAS $N_{3,3}$&    6-layer FFNN& 305   & - \\
$NN_9$&ACAS $N_{1,9}$&    6-layer FFNN& 305   & - \\

%Jigsaw&  RNN& LSTM&   8-cell&   0.9166\\
\bottomrule
\end{tabular}
\label{tab:dataset}
\end{table}

\subsection{Research Questions and Answers}\label{subsec:rq}

In the following, we report our experiment results and answer multiple research questions. \\

\noindent \emph{RQ1: Is CARE successful in neural network repair?}
To answer this question, we systematically apply CARE to the above-mentioned neural networks.
%\todo{Do we have baselines to compare with for each kinds of verification tasks? If there are, we should compare.}

For \emph{fairness repair tasks}, to better evaluate the performance of CARE, we set a strict fairness requirement (i.e., $\xi = 1\%$) to make sure all models fail the fairness property and CARE performs fairness repair on all of them (note that some of models fail fairness even if $\xi$ is set to be a realistic value of $5\%$). Table~\ref{tab:rq1} summarizes the results, where the columns show model, the protected feature, the unfairness (maximum probability difference) before and after repair and model accuracy before and after repair. Note that in this experiment, input features are assumed to be independent and normal distribution is followed in our sampling process. The number of neurons to repair is set to be no more than $10\%$ of total number of neurons in the given network. As illustrated in Table \ref{tab:rq1}, among all cases, $NN_2$ and $NN_3$ show alarming fairness concern (i.e., with unfairness above $5\%$). CARE successfully repairs all neural networks, with an average fairness improvement of $61.9\%$ and maximum of $99.2\%$. In terms of model accuracy, either the accuracy is unchanged ($NN_2$) or has only a slight drop ($NN_1$ and $NN_3$). %Note that we can adjust the parameter $\alpha$ to achieve better fairness or accuracy depending on the user requirement.

We further compare the performance of CARE with the state-of-the-art work~\cite{FM21} for this task. The method proposed in~\cite{FM21}, relies on learning a Markov Chain from original network and performing sensitivity analysis on it. Then optimization is applied to find small adjustments to weight parameters of sensitive neurons for better fairness. We run experiments on $NN_1$, $NN_2$ and $NN_3$ and the performance comparison is shown in Table~\ref{tab:rq5.comp} (DTMC represents the method proposed in~\cite{FM21}). CARE is able to improve fairness by $61.9\%$ on average and the model accuracy drops by $1.7\%$, while DTMC only improves fairness by $45.1\%$ at a higher cost of $2.2\%$. 
 
\begin{comment} 
\begin{figure}[!tbp]
  \centering
  \subfloat[Real]{\includegraphics[width=0.12\textwidth]{images/real.png}}
  \hspace{3em}%
  \subfloat[Reversed]{\includegraphics[width=0.12\textwidth]{images/re.png}}
\caption{Trigger comparison}
\label{fig:trigger}
\end{figure} 
\end{comment}
 
For \emph{backdoor removal tasks}, we train neural networks $NN_4$, $NN_5$ and $NN_6$ following the attack methodology proposed in BadNets~\cite{badnet}. We randomly chose a target label for each network and vary the ratio of adversarial inputs in training to achieve a high attack success rate of $>95\%$ while maintaining high accuracy. The trigger is a white square located at the bottom right of the image with size $4 \times 4$ (around $1\%$ of the entire image). In the experiment, we assume the trigger is unknown (as it is often the case in real application) and testing dataset is available. We adopt the technique proposed in~\cite{NC} to reverse engineer the trigger. As discovered in~\cite{NC}, the reverse-engineered trigger is not perfect in terms of visual similarity. However, CARE is able to remove the backdoor effectively as shown below.
%But according to~\cite{NC}, both reversed trigger and actual real trigger activate the same backdoor-related neurons. 
%Since it is often the case that the real trigger is unknown, we use reversed trigger in our experiment.
%\todo{Is the reverse-engineered trigger perfect? If not, we should discuss it here - it could be an advantage.} 
 Since $NN_4$, $NN_5$ and $NN_6$ are CNNs and convolutional layers tend to extract local features such as edges, textures, objects and scenes~\cite{cnn}, we apply CARE on dense layers only. We measure the attack success rate (SR) by adding the reverse-engineered trigger to all images in the testing set. As shown in Table~\ref{tab:rq1.backdoor}, for all the three models, CARE is able to mitigate the backdoor attack by reducing the attack SR from over $98\%$ to less than $1\%$, while the accuracy of the repaired networks either maintains the same ($NN_4$) , reduced by $<1\%$ ($NN_5$) or even improved ($NN_6$).
 
 Furthermore, we compare CARE with the state-of-the-art neural network repair technique proposed in~\cite{nnrepair} named NNRepair. NNRepair leverages activation map~\cite{activation} to perform fault localization of a buggy neural network and fix undesirable behaviors using constraint solving. We conduct experiments over CARE and NNRepair on two CNN models (experiment subjects of NNRepair) trained on MNIST~\cite{mnist} and CIFAR10~\cite{cifar10} datasets. The average performance comparison is shown in Figure~\ref{tab:rq5.comp}. CARE is able to reduces the SR by $99.9\%$ with accuracy drop of $1.5\%$ on average. In contrast, NNRepair only reduce SR by $19.3\%$ and model accuracy drops by $5.9\%$. 

For \emph{safety property repair tasks}, we use $N_{2,9}$, $N_{3,3}$ and $N_{1,9}$ of ACAS Xu networks as our $NN_7$, $NN_8$ and $NN_9$\footnote{A bug was found in the inital implementation for the sampling range of $NN_9$ and all the related results are updated accordingly after the fix}. These networks take 5-dimentional inputs of sensor readings and output 5 possible maneuver advises. Katz~\emph{et al.}~\cite{Reluplex} and Long~\emph{et al.}~\cite{Pham2020SOCRATESTA}, demonstrate that $N_{2,9}$ violates the safety property $\phi_{8}$, $N_{3,3}$ violates property $\phi_2$ and $N_{1,9}$ violates property $\phi_7$. Therefore, we apply CARE on these 3 networks to improve their performance on the corresponding property. In this experiment, for each network, we randomly sample 10K counterexamples to the property as the evaluation set and 10K instances that are correctly classified by the original neural network as the drawdown set. We measure the violation rate (VR) on each set to evaluate the performance. As shown in Table~\ref{tab:rq1.acas}, for all the three networks, CARE successfully brings down the violation rate from $1.0$ to $<1\%$ in evaluation set, while the performance in drawdown set is not affected.

To further evaluate the performance of CARE on safety property repair tasks, we compare CARE with the state-of-the-art approach proposed in~\cite{provable} named PRDNN. PRDNN~\cite{provable} introduces decoupled neural network architecture and solves neural network repair problem as a linear programming problem. 
%We follow the method proposed by~\cite{provable}, where we randomly sample 5,466 counterexamples to the safety property to measure violation rate and the same amount of samples that are classified correctly by the original network to measure the miss-classified rate after repair as the cost. 
We apply PRDNN on the above mentioned three ACAS Xu networks. The average result of $NN_7$, $NN_8$ and $NN_9$\footnote{PRDNN hangs with our initial experimental setup but we resolved this issue by setting up another Ubuntu virtual machine with more disk size.} is shown in Table~\ref{tab:rq5.comp}. Both tools are able improve the performance effectively while CARE outer-performs PRDNN by $2.2\%$ although the cost is slightly higher (still below $1\%$)\footnote{A bug was found in the initial code modifications to PRDNN for $NN_8$ and $NN_9$ and the results are updated accordingly after the fix.}\footnote{We use the same validation set generated by CARE to evaluate both tools.}.
 
 Thus to answer RQ1, we say CARE is able to repair given buggy neural network successfully for various properties while maintaining high model accuracy. Furthermore, CARE outer-performs existing works for different tasks.\\

 \noindent \emph{RQ2: What is the effect of parameter $\alpha$, i.e., the trade-off between the degree of property violation and accuracy?} Recall that the value of $\alpha$ controls the trade-off between fairness and accuracy. To understand the effect of $\alpha$'s value, we perform another set of experiments. 
% \todo{We need to have the data for all the models; and draw all of them in the same figure.}

Figure~\ref{fig:alpha} illustrates the result on all neural networks, where the first plot shows the performance improvement of repaired network compared to original network and the second plot shows the cost measured by model accuracy drop. As $\alpha$ increases from $0.1$ to $0.9$, the importance of model accuracy over repair objective in PSO fitness function increases. As shown in the plots for $NN_1$, the model accuracy is quite sensitive with respect to the value of $\alpha$, i.e., when $\alpha=0.2$ model accuracy drops by $0.05$ (from $0.88$ to $0.83$). Although a smaller $\alpha$ results in better fairness, we select a larger $\alpha$ to keep high model accuracy.
As shown in the plots for $NN_2$, model accuracy is stable over different $\alpha$ values. Similar to the previous case, smaller $\alpha$ results in better fairness. Therefore, we select a small $\alpha$ for more effective fairness repair. 
%We omit the results for other models due to space limitation.
As for $NN_3$, a large $\alpha$ improves model accuracy significantly. When $\alpha$ is greater than 0.7, the model accuracy is even higher than the original network. That is because PSO tries to optimize for model accuracy as well. However, the effectiveness of fairness improvement drops, i.e., model unfairness is reduced to 0.001 for $\alpha=0.1$ but at $\alpha=0.9$ the unfairness is as high as 0.0368. 
For $NN_4$, $NN_5$ and $NN_6$ costs for all $\alpha$ values are small ($< 1\%$).
%and a larger $\alpha$ further improves model accuracy.
This is because fixing the backdoor itself will improve model accuracy as well. 
For $NN_4$ and $NN_5$ the performance improvement (drop in SR) is quite stable over different $\alpha$. For $NN_6$, the performance improvement drops as $\alpha$ increases, hence, we select small $\alpha$ to let the optimization focus on backdoor removal.
As for $NN_7$, $NN_8$ and $NN_9$, the performance improvement is stable over different $\alpha$'s values. For $NN_7$, the cost drops from $0.25\%$ to $0.02\%$ when $\alpha$ increases from $0.1$ to $> 0.2$. While for $NN_8$ and $NN_9$, value of $\alpha$ does not affect the cost much. Hence, we select small $\alpha$ value for $NN_7$ and $NN_8$ and select $\alpha > 0.1$ for $NN_9$.
As shown by the experiment results, the value of $\alpha$ balances the trade-off between performance improvement and cost. A smaller $\alpha$ often results in more effective property repair but model accuracy may be affected.
%To this end, we carefully select the value of parameter $\alpha$ for each model, aiming to maximize the fairness improvement while keeping model accuracy difference to be within $3\%$. 
%In our experiments, we set $\alpha$ to 0.3 for models trained on Census Income and Bank Marketing datasets and set $\alpha$ to 0.1 for model trained on German Credit dataset. 
In our experiments, we set the value of $\alpha$ as shown in Table \ref{tab:alpha}.
%\todo{We should use the experiment results on the other models to justify the $\alpha$ value for the German Credit dataset}.
%\todo{One model is suspicious, we should have at least one model for each dataset for this experiment.} 

Thus to answer RQ2, we say that $\alpha$ does have a significant impact on the repairing result. In practice, the users can decide its value based on the importance of the property. In the future, we aim to study systematic ways of identifying the best $\alpha$ value automatically.\\
 
\begin{table}[t]
\centering
\caption{Fairness Repair Results}
\begin{tabular}[t]{p{8mm}|p{8mm}|p{9mm}p{9mm}|p{6mm}p{6mm}|p{4mm}p{4mm}}%{c|c|cc|cc|cc}
\toprule
\multirow{2}{*}{Model}&\multirow{2}{*}{P. Feat.}&\multicolumn{2}{c|}{Fairness Score}&\multicolumn{2}{c|}{Accuracy}&\multicolumn{2}{c}{Time(s)}\\
 & &Before&After&Before&After&Loc&Tot\\
\midrule
$NN_1$&Race&0.0431&0.0119&0.88&0.85&118&314\\
$NN_1$&Age&0.0331&0.0230&0.88&0.86&173&276\\
$NN_1$&Gender&0.0130&0.0070&0.88&0.86&120&240\\
%NN_2&Race&0.0020&0.0017&0.82&0.81&178&242\\
%NN_2&Age&0.0103&0.0011&0.82&0.81&289&403\\
%NN_2&Gender&0.0096&0.0061&0.82&0.81&215&309\\
%NN_3&Race&0.0117&0.0002&0.88&0.86&420&523\\
%NN_3&Age&0.0220&0.0149&0.88&0.87&496&565\\
%NN_3&Gender&0.0099&0.0013&0.88&0.88&397&546\\
$NN_2$&Age&0.0659&0.0005&1.0&1.0&170&221\\
$NN_2$&Gender&0.0524&0.0374&1.0&1.0&103&129\\
$NN_3$&Age&0.0544&0.0028&0.92&0.91&174&435\\
\bottomrule

\end{tabular}
\label{tab:rq1}
%   \begin{tablenotes}
%      \small
%      \item 
%    \end{tablenotes}
\end{table}

\begin{table}[t]
\centering
\caption{Backdoor Removal Results}
\begin{tabular}[t]{p{8mm}|p{9mm}p{7mm}|p{8mm}p{8mm}|p{4mm}p{4mm}}%{c|c|cc|cc|cc}
\toprule
\multirow{2}{*}{Model}&\multicolumn{2}{c|}{Attack SR}&\multicolumn{2}{c|}{Accuracy}&\multicolumn{2}{c}{Time(s)}\\
 &Before&After&Before&After&Loc&Tot\\
\midrule
$NN_4$&0.9808&0&0.9657&0.9654&80&493\\
$NN_5$&0.9868&0&0.9909&0.9867&330&541\\
$NN_6$&0.9977&0.0007&0.9136&0.9199&108&492\\
\bottomrule

\end{tabular}
\label{tab:rq1.backdoor}
%   \begin{tablenotes}
%      \small
%      \item 
%    \end{tablenotes}
\end{table}

\begin{table}[t]
\centering
\caption{Safety Property Repair Results}
\begin{tabular}[t]{p{8mm}|p{9mm}p{9mm}|p{6mm}p{6mm}|p{4mm}p{4mm}}%{c|c|cc|cc|cc}
\toprule
\multirow{2}{*}{Model}&\multicolumn{2}{c|}{Counterexample VR }&\multicolumn{2}{c|}{Drawndown VR}&\multicolumn{2}{c}{Time(s)}\\
 &Before&After&Before&After&Loc&Tot\\
\midrule
$NN_7$&1.0&0.0065&0.0&0.0&33&503\\
$NN_8$&1.0&0.0&0.0&0.0&31&168\\
$NN_9$&1.0&0.0001&0.0&0.0&39&325\\
\bottomrule

\end{tabular}
\label{tab:rq1.acas}
%   \begin{tablenotes}
%      \small
%      \item 
%    \end{tablenotes}
\end{table} 

\begin{table}[t]
\centering
\caption{Comparison with Existing Works}
\begin{tabular}[t]{c|cccccc}%{c|{12mm}c{6mm}c{6mm}c{6mm}c{6mm}c{6mm}c{6mm}}%{c|c|cc|cc|cc}
\toprule

\multirow{2}{*}{Tech.}&\multicolumn{2}{c}{Fairness}&\multicolumn{2}{c}{Backdoor}&\multicolumn{2}{c}{Safety}\\
 &Imp.&Cost&Imp.&Cost&Imp.&Cost\\
\midrule
CARE&61.9\%&1.7\%&99.9\%&1.5\%&99.6\%&0.15\%\\
DTMC&45.1\%&2.2\%&-&-&-&-\\
NNREPAIR&-&-&19.3\%&5.9\%&-&-\\
PRDNN&-&-&-&-&97.4\%&$0.01\%$\\
\bottomrule

\end{tabular}
\label{tab:rq5.comp}
%   \begin{tablenotes}
%      \small
%      \item 
%    \end{tablenotes}
\end{table}
 
\begin{table}[t]
\centering
\caption{Parameter $\alpha$ Setting}
\begin{tabular}[t]{p{1.0cm}p{0.6cm}p{0.6cm}p{0.8cm}p{0.6cm}p{0.8cm}p{0.8cm}}
\toprule
Model&$NN_1$&$NN_1$&$NN_1$&$NN_2$&$NN_2$&$NN_3$\\
P. Feat.&Race&Age&Gender&Age&Gender&Age\\
\midrule
$\alpha$&0.1&0.7&0.8&0.1&0.1&0.3\\
\toprule
Model&$NN_4$&$NN_5$&$NN_6$&$NN_7$&$NN_8$&$NN_9$\\
\midrule
$\alpha$&0.2&0.2&0.2&0.3&0.1&0.2\\
\bottomrule
\end{tabular}
\label{tab:alpha}
\end{table}
 
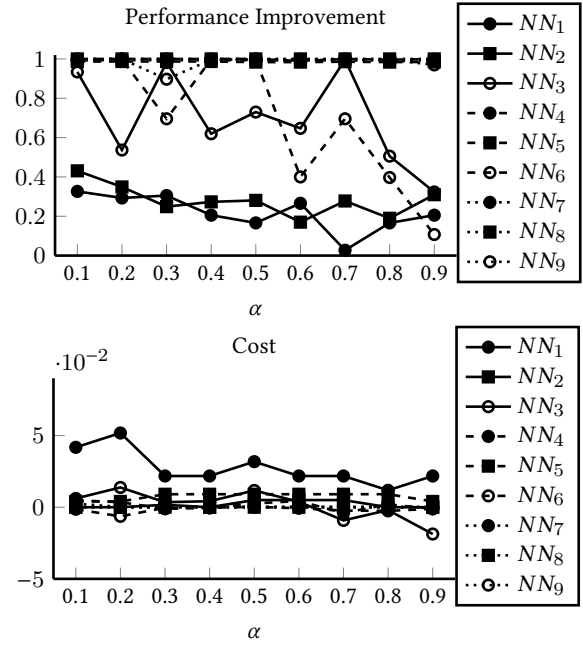
\begin{figure}
\begin{tikzpicture}
\pgfplotsset{%
    width=0.3\textwidth,
    height=0.15\textwidth
}

        \begin{axis}[%
                scale only axis,
                xmin = 0.05, xmax = 0.95,
                xtick={0.1,0.2,0.3,0.4,0.5,0.6,0.7,0.8,0.9},
                xticklabels={0.1,0.2,0.3,0.4,0.5,0.6,0.7,0.8,0.9},
                %xmajorgrids,
                ymin=0.0, ymax=1.02,
                xlabel={$\alpha$},
                %ymajorgrids,
                title={Performance Improvement},
                axis lines*=left,
                line width=1.0pt,
                mark size=2.0pt,
                legend style={at={(1.03,0.8)},
                %anchor= north east west,
                draw=black,fill=white,align=left},
                legend style={at={(7cm,1.5cm)},anchor=east}
                ]

                \addplot [
                color=black,
                solid,
                mark=*,
                mark options={solid},
                ]
                coordinates{
                    (0.1,0.326283987915408)
                    (0.2,0.293051359516616)
                    (0.3,0.305135951661631)
                    (0.4,0.205438066465257)
                    (0.5,0.166163141993958)
                    (0.6,0.265861027190332)
                    (0.7,0.0271903323262839)
                    (0.8,0.166163141993958)
                    (0.9,0.205438066465257)
                };
                \addlegendentry{$NN_1$};

                \addplot [
                color=black,
                solid,
                mark=square*,
                mark options={solid},
                ]
                coordinates{
                    (0.1,0.430955993930197)
                    (0.2,0.349013657056146)
                    (0.3,0.24886191198786)
                    (0.4,0.273141122913505)
                    (0.5,0.280728376327769)
                    (0.6,0.169954476479514)
                    (0.7,0.277693474962064)
                    (0.8,0.189681335356601)
                    (0.9,0.309559939301973)
                };
                \addlegendentry{$NN_2$};

                \addplot [
                color=black,
                %dash pattern=on 1pt off 3pt on 3pt off 3pt,
                solid,
                mark=o,
                mark options={solid},
                ]
                coordinates{
                    (0.1,0.933823529411765)
                    (0.2,0.536764705882353)
                    (0.3,0.981617647058823)
                    (0.4,0.619485294117647)
                    (0.5,0.729779411764706)
                    (0.6,0.647058823529412)
                    (0.7,0.996323529411765)
                    (0.8,0.505514705882353)
                    (0.9,0.323529411764706)
                };
                \addlegendentry{$NN_3$};

                \addplot [
                color=black,
                dashed,
                mark=*,
                mark options={solid},
                ]
                coordinates{
                    (0.1,1)
                    (0.2,1)
                    (0.3,0.999794440351523)
                    (0.4,0.999794440351523)
                    (0.5,1)
                    (0.6,0.989310898279219)
                    (0.7,1)
                    (0.8,0.996299926327422)
                    (0.9,0.998766642109141)
                };
                \addlegendentry{$NN_4$};

                \addplot [
                color=black,
                %dash pattern=on 3pt off 6pt on 6pt off 6pt,
                dashed,
                mark=square*,
                mark options={solid},
                ]
                coordinates{
                    (0.1,0.9868)
                    (0.2,0.9868)
                    (0.3,0.9868)
                    (0.4,0.9868)
                    (0.5,0.9838)
                    (0.6,0.9838)
                    (0.7,0.9858)
                    (0.8,0.98378)
                    (0.9,0.98478)
                };
                \addlegendentry{$NN_5$};
                
                              \addplot [
                color=black,
                dashed,
                mark=o,
                mark options={solid},
                ]
                coordinates{
                    (0.1,0.997027956989247)
                    (0.2,0.997027956989247)
                    (0.3,0.695616666666667)
                    (0.4,0.996691935483871)
                    (0.5,0.996691935483871)
                    (0.6,0.399917741935484)
                    (0.7,0.695611158117398)
                    (0.8,0.396893548387097)
                    (0.9,0.107243010752689)
                };
                \addlegendentry{$NN_6$};

                \addplot [
                color=black,
                dotted,
                mark=*,
                mark options={solid},
                ]
                coordinates{
                    (0.1,0.99291)
                    (0.2,0.9934)
                    (0.3,0.9935)
                    (0.4,0.9931)
                    (0.5,0.9929)
                    (0.6,0.9928)
                    (0.7,0.9928)
                    (0.8,0.9929)
                    (0.9,0.9927)
                };
                \addlegendentry{$NN_7$};
                
                \addplot [
                color=black,
                dotted,
                mark=square*,
                mark options={solid},
                ]
                coordinates{
                    (0.1,1)
                    (0.2,1)
                    (0.3,1)
                    (0.4,1)
                    (0.5,1)
                    (0.6,1)
                    (0.7,1)
                    (0.8,1)
                    (0.9,1)
                };
                \addlegendentry{$NN_8$};
                
                \addplot [
                color=black,
                dotted,
                mark=o,
                mark options={solid},
                ]
                coordinates{
                    (0.1,0.9956)
                    (0.2,0.9999)
                    (0.3,0.8967)
                    (0.4,0.9961)
                    (0.5,0.9978)
                    (0.6,0.9983)
                    (0.7,0.9941)
                    (0.8,0.9963)
                    (0.9,0.9695)
                };
                \addlegendentry{$NN_9$};

        \end{axis} 
\end{tikzpicture}

\begin{tikzpicture}
\pgfplotsset{%
    width=0.3\textwidth,
    height=0.15\textwidth
}

        \begin{axis}[%
                scale only axis,
                xmin = 0.05, xmax = 0.95,
                xtick={0.1,0.2,0.3,0.4,0.5,0.6,0.7,0.8,0.9},
                xticklabels={0.1,0.2,0.3,0.4,0.5,0.6,0.7,0.8,0.9},
                %xmajorgrids,
                ymin=-0.05, ymax=0.09,
                xlabel={$\alpha$},
                %ymajorgrids,
                title={Cost},
                axis lines*=left,
                line width=1.0pt,
                mark size=2.0pt,
                legend style={at={(1.03,0.8)},
                %anchor= north east west,
                draw=black,fill=white,align=left},
                legend style={at={(7cm,1.5cm)},anchor=east}]

                \addplot [
                color=black,
                solid,
                mark=*,
                mark options={solid},
                ]
                coordinates{
                    (0.1,0.0418000000000001)
                    (0.2,0.0518000000000001)
                    (0.3,0.0218)
                    (0.4,0.0218)
                    (0.5,0.0318000000000001)
                    (0.6,0.0218)
                    (0.7,0.0218)
                    (0.8,0.0118)
                    (0.9,0.0218)
                };
                \addlegendentry{$NN_1$};

                \addplot [
                color=black,
                solid,
                mark=square*,
                mark options={solid},
                ]
                coordinates{
                    (0.1,0)
                    (0.2,0)
                    (0.3,0.00170000000000003)
                    (0.4,0)
                    (0.5,0.005)
                    (0.6,0.005)
                    (0.7,0.005)
                    (0.8,0)
                    (0.9,0)
                };
                \addlegendentry{$NN_2$};

                \addplot [
                color=black,
                %dash pattern=on 1pt off 3pt on 3pt off 3pt,
                solid,
                mark=o,
                mark options={solid},
                ]
                coordinates{
                    (0.1,0.00609999999999999)
                    (0.2,0.0138)
                    (0.3,0.00350000000000006)
                    (0.4,0.00419999999999998)
                    (0.5,0.0117)
                    (0.6,0.00380000000000003)
                    (0.7,-0.00919999999999999)
                    (0.8,-0.002)
                    (0.9,-0.0186999999999999)
                };
                \addlegendentry{$NN_3$};

                \addplot [
                color=black,
                dashed,
                mark=*,
                mark options={solid},
                ]
                coordinates{
                    (0.1,0.00487096774193574)
                    (0.2,0.00350806451612984)
                    (0.3,-0.00112903225806427)
                    (0.4,-0.000322580645160886)
                    (0.5,0.000282258064516538)
                    (0.6,-0.000725806451612465)
                    (0.7,-0.00233870967741912)
                    (0.8,-0.00254032258064485)
                    (0.9,-0.00112903225806404)
                };
                \addlegendentry{$NN_4$};

                \addplot [
                color=black,
                %dash pattern=on 3pt off 6pt on 6pt off 6pt,
                dashed,
                mark=square*,
                mark options={solid},
                ]
                coordinates{
                    (0.1,0.00403229999999999)
                    (0.2,0.00403225806451701)
                    (0.3,0.00907258064516203)
                    (0.4,0.00907258064516203)
                    (0.5,0.00907000000000002)
                    (0.6,0.00907000000000002)
                    (0.7,0.00907000000000002)
                    (0.8,0.0091)
                    (0.9,0.00403299999999995)
                };
                \addlegendentry{$NN_5$};
                
                \addplot [
                color=black,
                %dash pattern=on 3pt off 6pt on 6pt off 6pt,
                dashed,
                mark=o,
                mark options={solid},
                ]
                coordinates{
                    (0.1,-0.00138655913978469)
                    (0.2,-0.00636628767847669)
                    (0.3,-0.00037849462365569)
                    (0.4,-0.000521716904635761)
                    (0.5,0.00264569892473154)
                    (0.6,0.00432580645161329)
                    (0.7,-0.0040086285915738)
                    (0.8,0.00197365591397891)
                    (0.9,-0.0010505376344081)
                };
                \addlegendentry{$NN_6$};

                \addplot [
                color=black,
                dotted,
                mark=*,
                mark options={solid},
                ]
                coordinates{
                    (0.1,0.00249999999999995)
                    (0.2,0.000199999999999978)
                    (0.3,0.0002)
                    (0.4,0.0002)
                    (0.5,0.0002)
                    (0.6,0.0002)
                    (0.7,0.0002)
                    (0.8,0.0002)
                    (0.9,0.0002)
                };
                \addlegendentry{$NN_7$};
                
                \addplot [
                color=black,
                dotted,
                mark=square*,
                mark options={solid},
                ]
                coordinates{
                    (0.1,0)
                    (0.2,0)
                    (0.3,0)
                    (0.4,0)
                    (0.5,0)
                    (0.6,0)
                    (0.7,0)
                    (0.8,0)
                    (0.9,0)
                };
                \addlegendentry{$NN_8$}; 
                
                \addplot [
                color=black,
                dotted,
                mark=o,
                mark options={solid},
                ]
                coordinates{
                    (0.1,0.0001)
                    (0.2,0)
                    (0.3,0)
                    (0.4,0)
                    (0.5,0.0002)
                    (0.6,0.0001)
                    (0.7,0.0001)
                    (0.8,0)
                    (0.9,0)
                };
                \addlegendentry{$NN_9$};

        \end{axis} 
\end{tikzpicture}
 
\caption{Effect of Parameter $\alpha$}
\label{fig:alpha}
\end{figure} 

\begin{table*}[t]
  \centering
  \begin{tabular}[t]{p{18mm}|p{10mm}p{10mm}p{10mm}p{10mm}|p{10mm}p{10mm}p{10mm}p{10mm}}
    \toprule
    \multirow{2}{*}{Model}&\multicolumn{4}{c|}{Performance Improvement}&\multicolumn{4}{c}{Cost}\\
     &CARE&Gradient&Random&All&CARE&Gradient&Random&All\\
    \midrule
$NN_1$ Race&$0.7222$&$0.9953$&$0.9953$&$1.0$&$0.0318$&$0.1218$&$0.1418$&$0.1318$\\
$NN_1$ Age&$0.3051$&$0.9486$&$0.9987$&$0.9848$&$0.0218$&$0.1818$&$0.1018$&$0.1418$\\
$NN_1$ Gender&$0.4615$&$0.6692$&$0.7384$&$0.9230$&$0.0218$&$0.1418$&$0.1118$&$0.1018$\\
$NN_2$ Age&$0.9924$&$0.1365$&$0.4324$&$0.0273$&$0.0$&$0.4000$&$0.1200$&$0.3500$\\
$NN_2$ Gender&$0.2847$&$0.7345$&$0.5989$&$0.1348$&$0.0$&$0.2000$&$0.1900$&$0.4200$\\
$NN_3$ Age&$0.9485$&$0.9301$&$0.6562$&$1.0$&$0.0125$&$0.0226$&$0.0426$&$0.0626$\\
$NN_4$&$1.0$&$0.0174$&$0.0050$&$0.3901$&$0.0002$&$0.0032$&$0.0050$&$0.0056$\\
$NN_5$&$0.9868$&$0.1058$&$0.0938$&$0.5220$&$0.0040$&$0.0019$&$0.0$&$0.0038$\\
$NN_6$&$0.9970$&$0.6995$&$0.0935$&$0.1729$&$-0.006$&$0.0016$&$0.0628$&$0.0034$\\
$NN_7$&$0.9935$&$0.0004$&$0.0207$&$0.5707$&$0.0$&$0.0$&$0.0$&$0.0001$\\
$NN_8$&$1.0$&$0.9250$&$0.6440$&$0.7485$&$0.0$&$0.0$&$0.0$&$0.0$\\
$NN_9$&$0.9999$&$0.1446$&$0.0663$&$0.9066$&$0.0$&$0.0$&$0.0$&$0.0002$\\
    \bottomrule
  \end{tabular}
  \caption{Fault Localization Effectiveness}
  \label{tab:rq3}
\end{table*}

\noindent \emph{RQ3: How effective is the causality-based fault localization?} 
%\todo{We need to compare the other repair tasks as well.}
This question asks whether our causality-based fault localization is useful. We answer this question by comparing the repair effectiveness through optimizing the weight parameters that are selected in four different ways, i.e., 1) selected based on our approach; 2) selected randomly; 3) selected based gradients and 4) include all the parameters. The results are shown in Table~\ref{tab:rq3}. 

Firstly,  we perform network repair with randomly selected neurons to optimize. We follow the same configuration of CARE as the one we used in RQ1, i.e., with the same $\alpha$ value and the number of neurons to fix. The performance improvement and cost of the repair is shown in \emph{Random1} columns in Table~\ref{tab:rq3}. Adjusting randomly selected parameters results in an average performance improvement of $47.6\%$. While CARE improves the performance by $80.8\%$ on all tasks. For fairness repair tasks, the improvement is significant with randomly selected neurons.
%In all nine models, the performance w.r.t. corresponding property improves after the repair. 
However, the model accuracy drops sharply. The overall accuracy drop is $8.9\%$ on average with the worst case of $19.0\%$ for $NN_2$ w.r.t. protected feature gender. For backdoor removal tasks, the performance improvement is $6.4\%$ on average and model accuracy is affected. Especially for $NN_6$, accuracy is reduced by $6.3\%$ after the repair. In terms of safety repair tasks, the repair is not so effective. For $NN_7$ the performance improvement is less than $1\%$ and in average the improvement is only $36.6\%$, although the model accuracy is not affected much. Therefore, with randomly selected parameters, the repair is not effective and the performance is improved at the cost of disrupting correct behavior of the original model. On the other side, as described in Section~\ref{sec:approach}, our fault localization is based on the causal influence of each neuron on the targeted misbehavior without altering other behaviors. 

%To further understand the performance of randomly selecting neurons, we adjust the value of $\alpha$ in each task for better performance. For fairness repair tasks,  we increase the parameter $\alpha$ to $0.8\sim0.9$ so as to ensure the model accuracy drop is minimized. As show in \emph{Random2} columns, the overall model accuracy drops by $0.04$ and the average fairness improvement is $33.7\%$. Fairness repair effectiveness is halved compared with CARE with causality-based fault localization. Specifically, for $NN_2$ w.r.t. protected feature gender and $NN_1$ w.r.t age, the fairness score is even worse than the original network. Note that for $NN_1$ w.r.t. race, we cannot get accuracy drop to be within $5\%$ even if we set $\alpha$ to 0.9. For backdoor removal tasks, we set $\alpha$ to $0.1$ for $NN_4$ and $NN_5$ to let the optimization focus on backdoor removal. We set $\alpha$ to $0.9$ for $NN_6$ to ensure the model accuracy is not affected in the repair. As shown in the table, the repair is still below $10\%$ for all the three models, although the model accuracy drop is not significant. For safety repair tasks, we set $\alpha$ to $0.1$ for all the models so as to ensure the performance improvement. However, the overall improvement is still around $40\%$ on average. Hence, for randomly selected parameters, the repair is not effective even with carefully selected $\alpha$ values.

Secondly, we optimize all weight parameters of the given neural network, i.e., no fault localization. In this setting, the search space in the PSO algorithm increases significantly (compared to the case of repairing $10\%$ of parameters in CARE). Therefore, we limit the time taken by PSO so that time allowed to spend in this step is the same as that taken by CARE with fault localization. The results are shown in columns \emph{All}. For fairness repair tasks the fairness is improved by $67.8\%$ on average. But the model accuracy drops significantly, i.e., the accuracy drops by $20.1\%$ on average and $42.0\%$ in the worst case. For backdoor removal, the performance improvement is only $36.2\%$ on average, while CARE reduces the SR by $99.5\%$. For safety repair tasks, the performance improvement is $55.6\%$ on average and CARE manages to reduce the VR in evaluation set by $99.8\%$. Although the cost is below $1\%$ for these two tasks, the repair is not that effective. In all of the cases without fault localization, PSO is not able to converge within the time limit and as a result, neither performance improvement nor cost is optimized when PSO terminates. 

In the literature, gradient is a commonly used technique that guides the search in program fault localization, program repair and fairness testing~\cite{gradient_penalty,arachne,adf2020}. Therefore, we conduct experiments to compare the performance of CARE with gradient-guided fault localization. That is, instead of calculating causal attribution of the hidden neurons, the fault localization step is based on $\frac{\partial y}{\partial v}$ where $y$ represents the favourable output value and $v$ represents hidden neuron value. We use the gradient-based method to identify the same amount of ``responsible" neurons and perform optimization with the same setting as CARE. The results of gradient-based localization method is illustrated in columns \emph{Gradient} of Table~\ref{tab:rq3}. Compared with CARE, for fairness repair tasks, the gradient-based method is able to find a repair that satisfies fairness criteria but fails to maintain model accuracy. Overall fairness improves by $73.6\%$ but average accuracy drops by $17.8\%$, which is not desirable. For backdoor removal, the gradient-based method is not able to find an effective repair where the overall performance improvement is below $30\%$. Especially for $NN_4$, the SR is still as high as $96.4\%$ after the repair. For safety repair, gradient-based method is able to find a repair for $NN_8$ and $NN_9$ with performance improvement above $90\%$, while CARE archives $100\%$ backdoor removal. For $NN_7$, gradient-based method is not useful at all where VR in evaluation set is as high as $99\%$ after the repair. Hence the performance of gradient-based method is not stable and often not effective.

Thus to answer RQ3, we say our causality-based fault localization is effective in identifying candidates for parameter optimization. It guides CARE to focus on fixing undesirable behavior of the model while keeping the correct behavior unaffected.\\

\section{Related Work} \label{sec:related}
This work is broadly related to works on neural network verification, repair and causal interpretation. \\

\noindent \emph{Neural network verification.} There have been an impressive line of methods proposed recently for neural network verification. These includes solving the problem using abstraction techniques~\cite{abstract20,deeppoly,ai2}, SMT sovler~\cite{marabou20,Reluplex,invariant20}, MILP and LP~\cite{Ehlers17,milp}, symbolic execution~\cite{symbolic18} and many others~\cite{bound20,ensemble20,exactness19,popqorn}. There have also been attempts to verify neural network fairness properties, including~\cite{fairsquare} and~\cite{verification} based on numerical integration ,~\cite{uai20} based on non-convex optimization and~\cite{FM21} based on probabilistic model checking. Unlike these works, our approach focus on neural network repair instead and we leverage some existing verification methods~\cite{FM21,NC,provable} in our work for property verification.\\

\noindent \emph{Machine learning model repair.} There have been multiple attempts on repairing machine learning models to remove undesirable behaviors. In~\cite{batchwise18}, Kauschke~\emph{et al.} suggest a method to learn a model to determine error regions of the given model, which allows users to patch the given model. In~\cite{patch19}, Sotoudeh \emph{et al.} propose a method to correct a neural network by applying small changes to model weights based on SMT formulation of the problem. 
In~\cite{minimal20}, Goldberger \emph{et al.} leverage recent advances in neural network verification and presented an approach to repair a network with minimal change in its weights. This approach aims to find minimal layer-wise fixes for a given point-wise specification and the performance is restricted by the underlying verification method used. NNRepair~\cite{nnrepair} performs constraint-based repair on neural networks to fix undesirable behaviors but only applies to a single layer. Sotoudeh~\emph{et al.} proposed a method in~\cite{provable} to solve a neural network repair problem as a linear programming problem. Unlike our approach, both NNrepair and the method proposed in~\cite{provable} perform layer-wise repair but the actual fault in a buggy network could be across multiple layers. Furthermore, based on our experiment results, our approach is more effective in neural network repair against different properties.\\

%Different from the above-mentioned studies, our approach focuses on repairing neural networks so that fairness is satisfied. While~\cite{fairrepair20,undertainty17} attempt to convert an unfair model into a fair one leveraging SMT solving techniques or based on input population guided repair, both of them are designed for simple decision making programs such as decision trees and random forests and hence do not apply to neural networks. %In contrast, our approach supports deep FFNN and RNN by design.

\noindent 
\emph{Machine learning causal interpretation.}
Causality analysis has been applied to generate explanations for machine learning algorithms. Existing works focus on causal interpretable models that can explain why machine learning models make certain decision. Narendra~\emph{et al.}~\cite{causal_dnn} model DNNs as SCMs and estimate the causal effect of each component of the model on the output. In ~\cite{cal_att}, Chattophadhyay~\emph{et al.} proposed an scalable causal approach to estimate individual causal effect of each feature on the model output. Causal inference has been applied in machine learning model fairness studies. Kusner~\emph{et al.}~\cite{cf_fairness} proposed an approach to measure the fairness of a machine learning model based on counterfactuals where a fair model should have the same prediction for both actual sample and counterfactual sample. In~\cite{causal_learning}, Zhang~\emph{et al.} propose a causal explanation based metric to quantitatively measure the fairness of an algorithm. Our work utilizes SCMs and $do(\cdot)$ calculus to measure the causal attribution of hidden neurons on model undesirable behaviors. The results are used as a guideline for fault localization. 

\section{Conclusion} \label{sec:conclusion}
We present CARE, a causality-based technique for repairing neural networks for various properties. CARE performs fault localization on a given neural network model and utilizes PSO to adjust the weight parameters of the identified neurons. CARE generates repaired networks with improved performance over specified property while maintaining the model's accuracy. CARE is evaluated empirically with multiple neural networks trained on benchmark datasets and experiment results show that CARE is able to repair buggy model efficiently and effectively, with minimally disruption on existing correct behavior. 

\section{Data Availability}
A Prove of Concept (PoC) realization of this work (CARE) is implemented on top of SOCRATES~\cite{Pham2020SOCRATESTA} as a causality-based neural network repair engine. The source code of CARE is publicly available at~\cite{Pham2020SOCRATESTA}.

\section{Acknowledgements} \label{sec:ack}
We thank anonymous reviewers for their constructive feedback. This research is supported by Huawei International (Grant Number TC20210714014).

\bibliography{ref}
\bibliographystyle{plain}

\end{document}